\documentclass[11pt,letterpaper]{article}

\usepackage[T1]{fontenc}
\usepackage[utf8]{inputenc}
\usepackage[english]{babel}
\usepackage{authblk,amssymb,amsfonts,amsthm,amsmath}
\usepackage{authblk,amssymb,amsfonts,amsmath}
\usepackage{xspace}
\usepackage{hyperref}
\usepackage{verbatim}
\usepackage{listings}

\usepackage{fullpage,times,paralist}
\let\olditem\item
\renewcommand{\item}{\setlength{\itemsep}{0pt}\setlength{\parskip}{0pt}\setlength{\parsep}{0pt}\setlength{\topsep}{0pt}\olditem}

\usepackage{tikz}
\usetikzlibrary {arrows,decorations,backgrounds}
\usepackage[caps]{complexity}
\usepackage{algorithmic}
\usepackage{algorithm}

\newtheorem{definition}{Definition}
\newtheorem{theorem}{Theorem}
\newtheorem{lemma}[theorem]{Lemma}
\newtheorem{fact}[theorem]{Fact}

\theoremstyle{plain}

\DeclareMathOperator{\Exp}{{\mathbb E}}
\DeclareMathOperator{\Reverse}{Reverse}
\DeclareMathOperator{\Sort}{Sort}

\newcommand{\eps}{\varepsilon}

\newcommand{\mi}{\mathrm{I}}
\newcommand{\h}{\mathrm{H}}
\newcommand{\Order}{\mathrm{O}}
\newcommand{\smallo}{\mathrm{o}}
\newcommand{\Omeg}{\mathrm{\Omega}}
\newcommand{\Thet}{\mathrm{\Theta}}
\renewcommand{\polylog}{\mathrm{polylog}}

\floatname{algorithm}{Algorithm}

\begin{document}

\title{Unidirectional Input/Output Streaming Complexity \\of Reversal
  and Sorting\thanks{Partially supported by the French ANR Blanc
    project ANR-12-BS02-005 (RDAM), the Singapore Ministry of
    Education Tier 3 Grant and also the Core Grants of the Centre for
    Quantum Technologies, Singapore.}}

\author[1]{Nathana\"el Fran\c{c}ois\thanks{\texttt{nathanael.francois@liafa.univ-paris-diderot.fr}}}
\author[2]{Rahul Jain\thanks{\texttt{rahul@comp.nus.edu.sg}}}
\author[3]{Fr\'ed\'eric Magniez\thanks{\texttt{frederic.magniez@cnrs.fr}}}
\affil[1]{Univ Paris Diderot, Sorbonne Paris-Cit\'e, LIAFA, CNRS, 75205 Paris, France}
\affil[2]{CQT and CS Department, National University of Singapore}
\affil[3]{CNRS, LIAFA, Univ Paris Diderot, Sorbonne Paris-Cit\'e, 75205 Paris, France}

\date{}

\maketitle

\begin{abstract}
We consider unidirectional data streams with restricted
access, such as read-only and write-only streams. For read-write
streams, we also introduce a new complexity measure called expansion,
the ratio between the space used on the stream and the input size.

We give tight bounds for the complexity of reversing a stream of length
$n$ in several of the possible models. In the read-only and write-only
model, we show that $p$-pass algorithms need memory space
$\Thet(n/p)$. But if either the output stream or the input stream is
read-write, then the complexity falls to $\Thet(n/p^2)$. It becomes
$\polylog(n)$ if $p=\Order(\log n)$ and both streams are read-write.

We also study the complexity of sorting a stream and give
two algorithms with small expansion. Our main sorting algorithm is
randomized and has $\Order(1)$ expansion, $\Order(\log n)$ passes and
$\Order(\log n)$ memory.
\end{abstract}

\section{Introduction}

\textbf{Background and Motivations.}
Streaming algorithms have been studied for 
estimating statistics, checking properties and computing functions
(more often with sublinear outputs) on massive inputs for several years.
However, less is known on computing functions that also have a massive
output, and therefore
require two data streams, one for the input and one for the output.
The notion of reversal complexity on a multi-tape machine,
which can be related to streaming complexity with multiple streams, was
first introduced in 1970 by Kameda and Vollmar~\cite{kv70} for
decision problems. The model however was not explored again until 1991
when Chen and Yap~\cite{cy91} considered the computable functions in this
model, and gave an algorithm for sorting using two streams with
$\Order(\log n)$ passes and (internal) memory space of size
$\Order(\log n)$. This was a significant improvement over the lower
bound for the case of a single stream which requires $\Omeg(n/s)$
passes, where $s$ is the size of the memory space, as proved by Munro
and Paterson~\cite{mp80}. 

Recently the interest for complexity in models with multiple streams
has been renewed. 
Indeed, streams are now considered as a model of external storage
allowing multiple sequential passes on them.
Hernich and Schweikardt~\cite{hs08} gave reductions
from classical complexity classes to the classes of problems decidable
with $\Order(1)$ streams. Grohe, Koch, Hernich and
Scwheikardt~\cite{gks07,ghs09} also gave several tight lower
bounds for multi-stream algorithms using $\smallo(\log n)$
passes. Gagie~\cite{g13} showed that two streams achieve perfect 
compression in polylogarithmic memory space and a polylogarithmic
number of passes. Beame, Huynh, Jayram and Rudra~\cite{bjr07,bh12}
also proved several
lower bounds on approximating frequency moments with multiple streams
and $\smallo(\log n)$ passes.
However, there is  a distinct lack of lower bounds for multi-stream algorithms 
with $\Omeg(\log n)$ passes. Indeed, most classical tools (such as
reduction to communication complexity) for proving lower bounds on
streaming algorithms fail for multiple streams. 
One exception is Ruhl's~\cite{Ruhl03} W-Stream and
StreamSort models. In W-Stream each (unidirectional) stream is
alternatively read-only and write-only, and passes are synchronized;
while this model seems similar to ours, it is actually much less
powerful and most lower bounds match the naive linear algorithms.
StreamSort is the W-Stream model augmented with a sorting
oracle and yields more interesting results, but of course trivializes
the problem we study in this paper.
Indeed, most classical tools (such as reduction to communication complexity)
for proving lower bounds on streaming algorithms fail for multiple streams.

%Instead of limiting the memory space to $\smallo(\log n)$, 
We therefore take an opposite approach and
restrict the number of streams to two,
the {\em input stream} and the {\em output stream}.
This leads to our notion of {\em input/output streaming algorithms}.
For many scenarios, this model is in fact more realistic since, as
an external storage, it comes at a cost to use several streams
simultaneously. In fact, the idea of using only input and output and
constrained secondary storage can also be found in other models, such
as \cite{av88}. Additionally, we only allow passes in one direction on
both streams.
This may looks too restrictive since bidirectional passes provide
exponential speedup in several decision problems, as exhibited by
several works~\cite{mmn10,jn10,cckm13,km13,fm13}.
However, changing the direction of processing a stream can have a higher
cost: most hard drives spin in only one direction and reading them in
the other is not as practical. 
We also consider three types of stream accesses:
{\em read-only} (for the input stream), {\em write-only} (for the output stream)
and {\em read-write} (for any or both streams). We emphasise that by
write-only stream we mean that once the algorithm has written in a
cell, it can never overwrite it.

For read-write streams, we introduce the new complexity measure of {\em expansion} of a
streaming algorithm, i.e. the ratio between the maximal size of the stream
during the computation and the size of the input (this concept does
not make sense for read-only and write-only stream, which cannot be
used for intermediate computations). While space on
external storage is by definition not as constrained as memory space,
in the case of massive inputs requiring large hard drives, 
expanding them even by a factor $\Order(\log n)$ is not reasonable. To
the best of our knowledge, while it has sometimes been implicitely
limited to $1$ or $\Order(1)$, this measure of complexity has never been
studied before for potential trade-offs with space or time.

We then study two problems, $\Reverse$ and $\Sort$,
in the model of input/output streaming algorithms.
Both problems consists in copying the content of the input stream to the output stream,
in reverse order for the first problem, and in sorted order for the second one.
%Problem $\Reverse$ consists in copying in reverse the content of the
%{input stream} to the {output stream};
%and problem $\Sort$ in copying in sorted order the input stream to the
%output stream.
Because of its importance in real applications, $\Sort$ 
has been extensively studied in many contexts. However, there
is still a lot to say about this problem in face of new models of
computation that massive data arise. 
$\Reverse$ can be seen as the simplest instance of $\Sort$.
It is natural to study its complexity since we forbid any reverse pass on streams.
%Doing so one can quantify the gap between algorithms with bidirectional passes and those without.
Moreover, any algorithm for $\Reverse$ gives a way to implement each of the efficient
bidirectional (and single stream) algorithms of~\cite{mmn10,km13,fm13} in our model,
thus giving another example of a speed-up with multiple streams. 

\smallskip

\textbf{Our results.}
%It appears that $\Reverse$ has the right complexity to gauge the improvement between our different models. 
In {Section~\ref{Reverse}}, we give tight bounds for the
complexity of $\Reverse$. We provide deterministic algorithms that are
optimal even against randomized ones.
Using communication complexity to prove lower bounds with multiple
streams is inherently difficult as soon as there are multiple passes
on both. Indeed, it is possible to copy one stream on the other one,
and since heads move separately,  the algorithm can access at any part
of the second stream while processing the first one. In the
read-only/write-only model, we instead skip communication complexity
and prove directly with information theory that a $p$-pass algorithm
needs space $\Omeg(n/p)$ (Theorem~\ref{theo_ROWO}).
In the read-only/read-write model, we similarly show that any $p$-pass
algorithm needs space $\Omeg(n/p^2)$
(Theorem~\ref{theo_RORW_lb}).
We provide an algorithm achieving that bound by copying large
(unreversed) blocks from the input, and
then placing correctly one element from each of the blocks during each pass
(Theorem~\ref{theo_RORW_alg}). Similar results hold in the
read-write/write-only model (Theorems~\ref{theo_RWWO_alg}
and~\ref{theo_RWWO_lb}). Last, we consider the read-write/read-write
model and give a $\Order(\log n)$-pass algorithm with polylogarithmic
memory space (Theorem~\ref{theo_RWRW}). 
All our algorithms presented 
%in Theorems~\ref{theo_RORW_alg}, \ref{theo_RWRO_alg} and \ref{theo_RWRW}. 
make extensive use of the ability to only write part
of a stream during a pass. In the more restrictive W-Stream model
where this is not possible, as shown by \cite{Ruhl03}, a $p$-pass
algorithm requires memory $\Omeg(n/p)$.

In {Section~\ref{Sort}}, we consider problem $\Sort$. % in our model.
Even with two streams only, it appears to have a tight bound, as
the algorithm by Chen and Yap~\cite{cy91} matches the lower bound
proved by Grohe, Hernich and Scwheikardt~\cite{ghs09}. However, Chen
and Yap's algorithm was not designed with expansion in mind: it works
by replicating the input $n$ times, thus using $\Order(n^2)$ cells on
the input and output streams. In this context, linear expansion of the
input is not reasonable. We therefore give two algorithms with two
streams that improve on it.
% although they do not quite match the known lower bound. 
The first one is deterministic. Similarly to the former algorithm, it
is based on Merge Sort, but has expansion $\Order(\log n)$ instead of
$\Omeg(n)$. It also uses $\Order(\log n)$ passes and space
$\Order(\log n)$ (Theorem~\ref{theo_MS}).
The second one is randomized and based on Quick Sort. It has
expansion $\Order(1)$, $\Order(\log n)$ passes and memory
 $\Order(\log n)$ (Theorem~\ref{theo_QS}). 
%Because adding a third stream makes the 
%problem trivial, if our algorithms were optimal it would give an
%factor $\Order(\log n)$ gap between two-streams and three-streams
%algorithms. We do not currently know of any such gaps, although many
%exponential gaps exist between one-stream and two-streams
%algorithms. It is worth noting that we still match lower bound on
%number of passes and memory space for multi-stream bidirectional algorithms.

\section{Preliminaries}
In streaming algorithms (see~\cite{MuthuBook} for an introduction), a
{\em pass\/} on a string $w\in\Sigma^n$, for some finite alphabet $\Sigma$,
means that $w$ is given as a {\em stream} $w[1],w[2],\ldots,w[n]$, which
arrives sequentially, i.e., letter by letter in this order. Depending
on the model, $w$ may or may not be modified.
For simplicity, we always assume $|\Sigma| = \Order(n)$, except when explicitly stated.
Otherwise, one should transpose our algorithms such that the letters are also read sequentially, say bit by bit.
%If $\Sigma$ is larger, we can consider $\Sigma' = \{w[i]\}_{1 \leq i \leq n}$ instead, which  is of size at most $\Order(\log n)$
We now fix in the rest of the paper such a finite alphabet $\Sigma$ equipped with a total order.

We consider {\em input/output streaming algorithms} with two data
streams $X$ and $Y$. 
Stream $X$ initially contains the input, and stream $Y$ is initially
empty and  will contain the output at the end of the execution of the
algorithm. We denote by $X[i]$ the $i$-th cell of stream $X$, and
similarly for $Y$. 
The size of a stream is the number of its cells containing some data.
Our algorithms have also access to a random access memory $M$, usually
of sublinear size in the input size. 

We then parameterize such  algorithms by the operations allowed on
the input stream $X$ (read-only or read-write), on the output stream
$Y$ (write-only or 
read-write), the number of passes, the bit-size of the memory $M$ and the
expansion of the streams with regard to the size of the input. 
Whenever we mention the memory of an algorithm, we mean a random-access memory.

\begin{definition}[Input/Output streaming algorithm]
Let I be either RO or RW, and let O be either WO or RW.
Then a {\em $p(n)$-pass I/O streaming algorithm} with space $s(n)$
is an algorithm that, given $w\in\Sigma^n$ as a stream $X$, produces
its output on an initially empty stream $Y$ and such that:
\begin{itemize} 
\item It performs $p(n)$ sequential passes in total on $X$ and $Y$;
\item It maintains a memory $M$ of size at most $s(n)$ bits while
  processing  $X,Y$; 
\item If I is RO, $X$ cannot be modified;
\item If O is WO, $Y$ cannot be read and each cell %$Y[i]$ 
of $Y$ can
  be written only once. 
\end{itemize}
Moreover, the algorithm has expansion $\lambda(n)$ if streams $X,Y$
have length (in number of cells) at most $\lambda(n) \times n$ during
its execution, 
for all input $w\in\Sigma^n$.
\end{definition}
Observe that we do not mention any running time in our
definition. Indeed, all our lower bounds will be stated independently
of it. Moreover, %even if we will not mention it explicitly, 
all our algorithms process each letter from each stream in
polylogarithmic time. They also have also polylogarithmic preprocessing and postprocessing times.

A usual way to get an algorithm with expansion $>1$ consists in annotating streams
with a larger alphabet $\Sigma'$, while keeping  (artificially) the same number of annotated cells.
In that case, one can simulate one annotated cell using $({\log|\Sigma'|})/{\log |\Sigma|}$
non-annotated cells.

For simplicity, we assume further that the input length $n$ is
always given to the algorithm in advance. Nonetheless, all our
algorithms can be adapted to the case in which $n$ is unknown until
the end of a pass.  
Moreover, this assumption only makes our lower bounds stronger.

We will use the two usual notions of randomized computing.
\begin{description}
\item[Monte Carlo:]
An algorithm computes a function $f$ on $\Sigma^n$ {\em with error
  $\eps\leq 1/3$} if  
for all inputs $w \in \Sigma^n$, it outputs $f(x)$ with
probability at least $1-\eps$.
\item[Las Vegas:]
An algorithm computes a function $f$ on $\Sigma^n$ {\em with
  failure $\eps\leq 1/2$} if  
for all inputs $w \in \Sigma^n$, it outputs $f(x)$ with
probability at least $1-\eps$,
otherwise it gives no output.
\end{description}

The two functions we study in this paper in term of streaming complexity are now formally defined.
%We are going to study the streaming complexity of the computation of the following two functions.
%The two functions we are going to compute are now defined formally.
%In this paper we consider the problems to compute one of the following two functions. 
\begin{definition}[$\Reverse$ and $\Sort$]
For a sequence $w=w[1]w[2]\ldots w[n] \in \Sigma^n$, let us define
$\Reverse(w)$ as $w[n]w[n-1]\ldots w[1]$.
When $\Sigma$ has a total order,
also define $\Sort(w)$ as the sorted permutation of $w$.
\end{definition}

For simplicity when proving lower bounds, for the proofs of
Theorems~\ref{theo_ROWO}, \ref{theo_RORW_lb} and \ref{theo_RWWO_lb}, 
instead of considering algorithms for $\Reverse$ writing $Reverse(X)$ on
$Y$ with passes from left to right, we will have them write $X$ on $Y$ with
passes from right to left, which is equivalent with a relabeling of
$Y$. The algorithms presented for upper bounds, however, read better
in the model where we write $\Reverse(X)$ on $Y$ using a pass from
left tor right.

When we describe a streaming algorithm using pseudo-code, a `For' loop
will always correspond to a single pass on each stream, except when
explicitly mentioned otherwise. Most algorithms will consist of
a `While' loop including a constant number of `For' loops,
i.e. passes. Therefore the number of iterations of the `While' loop
will give us the number of passes.

Throughout the paper, we always give  randomized lower bounds. While
communication complexity arguments fail in the context of mutliple
streams, we use information theory arguments. They have recently been
proved to be a powerful tool for proving lower bounds in
communication complexity, and therefore for streaming algorithms.
Here we not only use them directly, but also there is no direct
interpretation in term of communication complexity.
%Nevertheless, we use some information theory arguments that are usually the underlying concepts of entropy $\h$ and mutual information $\mi$. 

Let us remind now the notions of entropy $\h$ and
mutual information $\mi$. Let $X,Y,Z$ be three random variables. Then :
$$\quad \quad \quad \h(X)  =  -\Exp_{x\gets X}\log \Pr(X=x),$$
$$\quad \quad \quad \h(X|Y=y)  = -\Exp_{x\gets X}\log \Pr(X=x|Y=y),$$
$$\h(X|Y)  =  \Exp_{y\gets Y} \h(X|Y=y),$$
$$\mi(X:Y|Z)  = \h(X|Z)-\h(X|Y,Z).$$

The entropy and the mutual information are non negative
and satisfy $\mi(X:Y|Z)=\mi(Y:X|Z)$.
For $0\leq\eps\leq 1$, we denote by $\h(\eps)$ the entropy of the
Bernoulli variable
that takes value $1$ with probability $\eps$, and $0$ with probability $1-\eps$.

\section{Reversal with two streams}
\label{Reverse}
\subsection{Complexity in the Read-only/Write-only model}
When $\Sigma=\{0,1\}$, the naive algorithm, that copies at each pass
$s$ bits of the input in memory and then to the output in reverse order,
requires memory space $ s = {n}/{p}$ when performing $p$ passes. 
When processing each stream for the first time, it is in fact
obvious that $Y[i]$ cannot be written until $X[n-i+1]$ has been
read. In particular, for $p=2$ we have $s \geq n$ for any algorithm.
With multiple passes on each stream, the proof gets more technical
especially since the point where the streams cross, i.e. where we
can write on the second stream what was read on the first one during
the current pass depends on the pass, the input and the
randomness. However the $s \geq \Order({n}/{p})$ bound still applies for all
$p$, when the error probability of the algorithm is $\Order(1/p)$. 

Observe that a constant error probability can be reduced to $\Order(1/p)$
using $\Order(\log p)$ parallel repetitions, leading to an
$\Order(\log p)$ factor in the size of the memory space.

\begin{theorem}
Let $0 \leq \eps \leq 1/10$. Every $p$-pass randomized RO/WO algorithm 
    for $\Reverse$ on $\{0,1\}^n$ with error $\eps$ requires space
    $\Omeg ({n}/{p})$.
\label{theo_ROWO}
\end{theorem}

Expansion is not mentionned in this theorem as it does not make sense
on streams that are not read-write: the algorithm either cannot modify
the stream if it is read-only, or has no use for the extra data if it
is write-only.

Before proving Theorem~\ref{theo_ROWO}, we begin with two useful facts.
\begin{fact}
\label{E-H}
Let $X$ be uniformly distributed over $\{0,1\}^n$, and
let $J$ be some random variable on $\{0,1,\dots,n\}$ that may depend on $X$.
Then :
\vspace{-1mm}
$$\h(X[1,J]|J) \leq \Exp(J) \text{ and } \h(X[1,J]|X[J+1,n])  \geq
\Exp(J) - \h(J).$$
\vspace{-1mm}
Similarly, 
$$\h(X[J+1,n]|J)\leq n - \Exp(J) \text{ and } \h(X[J+1,n]|X[1,J]) \geq n - \Exp(J) - \h(J).$$
\end{fact}
%We begin with a proof of Fact~\ref{E-H}.
\begin{proof}
$\h(X[1,J]|J) \leq \Exp(J)$ and $\h(X[J+1,n]|J) \leq n - \Exp(J)$ are
  direct. The second part uses the first one as follows:
\vspace{-1mm}
\begin{eqnarray*}
\h(X[1,J]|X[J+1,n]) & = & \h(X|J,X[J+1,n]) - \h(X[J+1,n]|J,X[J+1,n]) \\
& = & \h(X|J) - \h(X[J+1,n]|J) \\
& \geq & \h(X) - \h(J) - n + \Exp(J) = \Exp(J) - \h(J).
\end{eqnarray*}
\end{proof}

\begin{fact}[Data processing inequality]
\label{DTI}
Let $X,Y,Z,R$ be random variables such that $R$ is independent from $X,Y,Z$.
Then  for every function $f$
$$\h(X | Y,Z) \leq \h(f(X,R) | Y,Z)\quad \text{and}\quad \mi(X:Y | Z)
\geq \mi(f(X,R):Y | Z).$$ 
\end{fact}
Note that the previous property is usually stated with no variable
$R$, then $f$ is defined as a {\em probabilistic function}.

\begin{proof}[Proof of Theorem~\ref{theo_ROWO}]
In this proof,
we use the equivalent model of the algorithm copying $X$ (and not
$\Reverse(X)$) on $Y$, but processing $Y$ only with passes from right
to left. This means that the two heads start on opposite ends and meet 
once each pass.

Consider a $p$-pass randomized RO/WO algorithm for
$\Reverse$ on $\{0,1\}^n$ with error probability $\eps$ and space $s\geq\log n$.
For simplicity, we assume that passes are
synchronized : whenever a pass on one stream ends, the head on the
other stream ends its own pass, and then eventually moves back to its original
position. This costs us at most a factor $2$ in the number of
passes. 

Let the input stream $X$ be uniformly distributed in $\{0,1\}^n$.
%Therefore, there are at most $p$ synchronous passes.
For each pass $1 \leq t \leq p$, let $Z^t \in 
\{0,1,\bot\}^n$ be the reverse of the data written on the output stream $Y$: if
nothing is written at pass $t$ and index $i$, then $Z^t[i] =
\bot$. This model corresponds to writing the same string on the output
stream as on the input stream, but going in different directions, an
makes the notations simpler as we want $X[i]=Z^t[i]$. Note
that because the algorithm cannot overwrite a letter, for each
$i$ there is at most one $t$ such that $Z^t[i] \not = \bot$.
Last, let  $L^t$ be the index where the reading
head and the writing head meet during pass $t$. 
Since passes are synchronized, $L^t$ is unique (but possibly depends
on the input and random choices). For $1 \leq i \leq n$, $1 \leq t
\leq p$, let $M_i^t$ be the state of the memory after the algorithm
reads $X[i]$ on pass $t$.

For  $1 \leq t \leq p$, since $s$ bounds the size of memory,
we have :  
$$s \geq \mi(X[1,L^t]:M_{L^t}^t|X[L^t+1,n])) \text{ and } s \geq
\mi(X[L^t+1,n]:M_{n}^{t+1}|X[1,L^t])).$$
Using the definition of mutual information, we get the following inequalities:
$$s \geq \h(X[1,L^t]|X[L^t+1,n]) - \h(X[1,L^t]|M_{L^t}^t,X[L^t+1,n]),$$
$$s \geq \h(X[L^t+1,n]|X[1,L^t]) - \h(X[L^t+1,n]|M_{n}^{t-1},X[1,L^t]).$$

We define the following probabilities : 
$q_i^t(l) = \Pr(Z^t[i] \not =\bot | L^t=l)$,
$q_i^t=\Pr(Z^t[i] \not =\bot)$, 
$\eps_i^t(l) = \Pr(Z^t[i] \not =\bot, Z^t[i] \not = X[i] | L^t = l)$ and 
$\eps_i^t = \Pr(Z^t[i] \not =\bot, Z^t[i] \not = X[i])$. By definition, they also satisfy
$\eps_i^t = \Exp_{l \sim  L^t}(\eps_i^t(l))$ and 
$q_i^t = \Exp_{l \sim  L^t}(q_i^t(l))$.
Note that by hypothesis\footnote{Here we only need the hypothesis that
  each bit of $Y$ is wrong with probability at most $\eps$, and not
  the stronger hypothesis that $Y \not = \Reverse(X)$ with probability
  at most $\eps$.} and because there is no rewriting, 
$\sum_{t=1}^p \eps_i^t \leq \Pr(\exists t, Z^t[i] \not= \bot, Z^t[i] \not = X[i]) \leq \eps$.
Lemmas~\ref{lem s(n)} and~\ref{lem H(e/q)} give us
these inequalities: 
$$2 s \geq n - \sum_{i=1}^n \h(X[i]|Z^t[i],L^t) - O(\log n),$$ 
$$\h(X[i]|Z^t[i],L^t) \leq 1 - q_i^t \left(1-\h ({\eps_i^t}/{q_i^t} \right) ).$$ 
Combining them yields :
$$2s \geq \sum_{i=1}^nq_i^t \left(1-\h\left(({\eps_i^t}/{q_i^t}) \right)\right) - O(\log n).$$

Let $ \alpha_i = \sum_{t=1}^p q_i^t$. Then $\alpha_i = \Pr[Y[i] \not =
  \bot]$ satisfies $\alpha_i \geq 1 - \eps$ by hypothesis.
Now summing over all passes leads to
$2ps  \geq  \sum_{i=1}^n \alpha_i \sum_{t=1}^p
(q_i^t/\alpha_i) (1-\h (\eps_i^t/q_i^t)) - \Order(p\log n).$

The concavity of $\h$ gives us $ \sum_{t=1}^p ({q_i^t} / {\alpha_i}) \h
({\eps_i^t}/{q_i^t} ) \leq \h ({\eps}/{(1-\eps)})$. This means, replacing
$\alpha_i$ and $\eps_i^t$ by their upper bounds, that
$ 2ps \geq n(1-\eps) (1-\h ({\eps}/{(1-\eps)} )) - \Order(p\log n)$. Since
Theorem~\ref{theo_ROWO} has $\eps \leq 0.1$ as an hypothesis,
our algorithm verifies $ps \geq \Omeg(n)$.
\end{proof}

\begin{lemma}
Assuming the hypotheses of Theorem~\ref{theo_ROWO}, at any given pass
$t$, $$2s \geq n - \sum_{i=1}^n \h(X[i]|Z^t[i],L) - O(\log n).$$
\label{lem s(n)}
\end{lemma}

\begin{proof}
In this proof, we write $Z[i]$ for $Z^t[i]$ since there is generally no
ambiguity. We similarly omit the $t$ on other notations.

The data processing inequality (Fact~\ref{DTI}) gives us the following
inequality :
$\h(X[1,L]|M_L,X[L+1,n],L) \leq \h(X[1,L]|Z[1,L],L)$. We can 
rewrite this as 
$$\h(X[1,L]|M_L,X[L+1,n],L) \leq \Exp_{l \sim  L}(\h(X[1,l]|Z[1,l],L=l)).$$ 
Using the chain rule and removing
conditioning, we get 
$$\h(X[1,L]|M_L,X[L+1,n],L) \leq \Exp_{l \sim L}(\sum_{i=1}^l \h(X[i]|Z[i],L=l)).$$
Similarly, 
$$\h(X[L+1,n]|M_n^{t-1},X[1,L],L) \leq
\Exp_{l \sim L} (\sum_{i=1}^l \h(X[i]|Z[i],L=l)).$$
Using that both $M_{L}$ (where $M_L=M^t_{L^t}$) and $M_{n}^{t-1}$ are
of size at most $s$ bits, we get 
$$2 s \geq \mi(X[1,L]:M_{L^t}^t|X[L+1,n]) + \mi(X[L+1,n]:M_{n}^{t-1}|X[1,L]).$$

Then we conclude by combining the above inequalities and using Fact~\ref{E-H} as follows:
\begin{eqnarray*}
2 s & \geq & \Exp(L) - \h(L) + n - \Exp(L) - \h(L) \\ && -
\h(X[1,L]|M_{L}^t,X[L+1,n]) - \h(X[L+1,n]|M_{n}^{t-1},X[1,L])\\
 & \geq & n - \Exp_{l \sim L} \left(\sum_{i=1}^l
\h(X[i]|Z[i],L=l) + \sum_{i=l+1}^n
\h(X[i]|Z[i],L=l) \right) - \Order(\log n)\\
 & = & n - \sum_{i=1}^n \h(X[i]|Z[i],L) - O(\log n).
\end{eqnarray*}
\end{proof}

\begin{lemma}
Assuming the hypotheses of Theorem~\ref{theo_ROWO}, for any pass $t$,
$\h(X[i]|Z^t[i],L_t) \leq 1 - q_i^t (1-\h (\eps_i^t/q_i^t))$
\label{lem H(e/q)}
\end{lemma}

\begin{proof}
As above, we omit the $t$ in the proof, as there is no ambiguity.

The statement has some similarities with Fano's inequality. Due to the specificities of our context, we have to revisit its proof as follows. First we write
\begin{eqnarray}
\h(X[i]|Z[i],L) & \leq & \Exp_{l \sim L}(\h(X[i]|Z[i],L=l)) \nonumber \\
& \leq & \Exp_{l \sim L}(q_i(l)\h(X[i]|Z[i],L=l,Z[i] \not = \bot)
\nonumber \\
 & & \quad + (1-q_i(l))\h(X[i]|L=l,Z[i] = \bot) \nonumber \\
& \leq & \Exp_{l \sim L} \left( q_i(l)\h \left(
\frac{\eps_i(l)}{q_i(l)}\right) + 1-q_i(l) \right)\nonumber \\
& = & 1 - q_i + \Exp_{l \sim L} \left( q_i(l) \h \left(\frac{\eps_i(l)}{q_i(l)}\right)\right) \label{H(X|Z)}.
\end{eqnarray}

By replacing the entropy with its definition, we can see that for any
$1\geq q \geq \eps>0$, we have  
$q\h\left(\frac{\eps}{q}\right)= \h(q-\eps,\eps, 1-q) - \h(q)$, 
where $\h(x,y,z)$ is the entropy
of a random variable $R$ in $\{0,1,2\}$ with $\Pr(R=0)=x$, $\Pr(R=1)=y$
and $\Pr(R=2)=z$. Let $R_i$ be such that $R_i=0$ if $X[i]=Z[i]$, $R_i=2$
if $Z[i] = \bot$ and $R_i=1$ otherwise.
Note that $(Z[i] = \bot)$ is a function of $R_i$.
Therefore:

\begin{eqnarray*}
\Exp_{l\sim L} \left( q_i(l)\h
\left(\frac{\eps_i(l)}{q_i(l)}\right)\right) & = & \Exp_{l \sim
  L}(\h(R_i|L=l) - \h((Z[i] = \bot)|L=l))\\
& = & \h(R_i|L) -
\h((Z[i] = \bot)|L)\\
& = & \h(R_i) - \h((Z[i]=\bot)) + \mi((Z[i]=\bot)|L) - \mi(R_i|L).
\end{eqnarray*}

By the data processing inequality, $\mi((Z[i]=\bot):L) \leq
\mi(R_i:L)$, so 
$$\Exp_{l\sim L} (q_i(l)\h (\eps_i(l)/q_i(l)))) \leq q_i\h
(\eps_i/q_i)).$$
Combining this with inequality \ref{H(X|Z)} gives us the lemma.
\end{proof}

\subsection{Complexity of Read-only/Read-write and Read-Write/Write-Only}

In this section, we prove tight bounds in the Read-only/Read-Write and
Read-Write/Write-Only model. 
These models are more complex than the
previous one since an algorithm may now modify $Y[i]$ or $X[i]$
several times, and use that as additional memory.

\smallskip

\textbf{Algorithm.}
As a subroutine we use Algorithm~\ref{RORWalg}, which performs $\Order(\sqrt{n})$ passes
and uses space $\Order(\log|\Sigma|)$. It works by copying blocks of size $\Order(\sqrt{n})$ 
from the input stream to the output stream 
without reversing them (otherwise there would not be enough space in
the memory), but putting them in the correct order pairwise. In
addition, during each pass on the output, it moves one element from
each block already copied to the correct place. Since blocks have as
many elements as there are passes left after they are copied, the output stream is in the correct order at the end of the execution .

\begin{lstlisting}[caption={RO/RW streaming algorithm for $\Reverse$},label=RORWalg,captionpos=t,float,abovecaptionskip=-\medskipamount,mathescape]
$p \gets \sqrt{2n}$, $t \gets 1$, $i_1 \gets p$;
While {$t \leq p$} (*@\label{RORW_loop}@*)
  If $t > 1$ then $(R,l) \gets (Y[n-i_t],n-i_t)$
  If $i_t < n$ then
    $Y[n-i_t-(p-t),n-i_t] \gets X[i_t-(p-t),i_t]$ // Order is unchanged (*@\label{copy}@*)
    $i_{t+1} \gets i_t+p-t$
  For $m=t-1$ to $1$ (*@\label{RORW_pass}@*)
    Put $R$ in the right place $Y[l']$ // $l'$ computed from $l,m,p,n$ (*@\label{Y[l]}@*)
    $(R,l) \gets (Y[l'],l')$
  $t \gets t+1$
\end{lstlisting}

\begin{theorem}
\label{theo_RORW_alg}
There is a deterministic algorithm such that, given $n$ and 
$p\leq\sqrt{n}$,
it is a $p$-pass RO/RW streaming algorithm for $\Reverse$ on
$\Sigma^n$ with space $\Order(\log n+(n\log|\Sigma|)/p^2)$ and expansion $1$.
\end{theorem}

\begin{proof} 
We will prove that Algorithm~\ref{RORWalg} satisfies the theorem when
$p=2\sqrt{n}$.
For the general case, the algorithm treats
groups of $m=4n/p^2$ letters as though they were just one letter of
the new alphabet $\Sigma^m$, 
%where $m$ is the smallest integer such that $n/m \geq |\Sigma|^m$.
and then runs Algorithm~\ref{RORWalg}. This uses spaces $\Order (\log(n/m)+(n/m) (\log|\Sigma^m|)/p^2) = \Order (\log n+ (n\log|\Sigma|)/p^2)$.

We now prove the theorem for $p=2\sqrt{n}$. First, observe that every element of the input is copied on the output, as $
\sum_{t=1}^{p}(p-t) = {p(p+1)}/{2} = n+ \sqrt{n} \geq n$.

Let $1 \leq t < p$.  The subword $X[i_t-(p-t)-1,i_t]$ is initially
copied at line~\ref{copy} on
$Y[n-i_t-(p-t)-1,n-i_t]$. Therefore only $Y[n-i_t]$ is correctly placed.
%For $t \geq 1$, 
Let $B_t$ be $\{n-i_t-(p-t),n-i_t-1\}=
\{n-i_{t+1}+1,n-i_t-1\}$. %, with $i_{0}=0$. 
Then $B_t$ denotes the indices (on $Y$) of elements copied during
the $t$-th iteration of the While loop that are incorrectly
placed. For each $l \in B_t$, the correct place for $Y[l]$ is in
$\{n-i_t+1,n-i_t+p-t\} =
\{n-i_t+1,n-(i_{t-1}+1)\}=B_{t-1}$, where by convention $B_0$ is defined with $i_0=0$. 
Therefore, in the $(t+1)$-th iteration
of the while loop, the first $Y[l]$ we place correctly goes from
$l=n-i_{t+1}+1=n-i_t-(p-t) \in B_{t}$ to some $l' \in B_{t-1}$. 
%This $Y[l']$ has not been replaced by the correct value before, since it is the first time we move $Y[n-i_t+1]$, and therefore has never been replaced at all. 
Then recursively the previous value of $Y[l']$ goes in $B_{t-2}$, and so on until we reach
$B_0$ where nothing was written initially. 

Thus, the For loop places correctly one element of each of $B_{t-1}$, $B_{t-2}, \dots,
B_1$. Observe that $B_m$ has at most $p-m$ elements incorrectly placed. Moreover,
there are $(p-m)$ remaining iterations of the While loop after $B_m$ is written.
Therfore all elements are places correctly when Algorithm~\ref{RORWalg} ends.

Now we prove that Algorithm~\ref{RORWalg} has the claimed
complexity. It only starts a new pass when $t$ increases, at each
execution of the While loop at line~\ref{RORW_loop}.
Indeed, each execution of line~\ref{Y[l]} can be performed within the current pass since
 $Y[l]$ moves forward to a new index in $B_{m-1}$ as explained before. 
Therefore %the For loop at line~\ref{RORW_pass} is indeed just one pass on each stream, and 
Algorithm~\ref{RORWalg} runs
in $p$ passes. Since it keeps at most two elements in memory at any
time, and only needs to keep track of the current position, $n$, $p$, $t$ and
$m$, it uses $O(\log n + \log|\Sigma|)$ memory.
\end{proof}

\begin{theorem}
\label{theo_RWWO_alg}
There is a deterministic algorithm such that, given $n$ and $p\leq\sqrt{n}$,
it is a $p$-pass RW/WO streaming algorithm for $\Reverse$ on
$\Sigma^n$ with space $\Order(\log n+(n\log|\Sigma|)/p^2)$ and expansion $1$.
\end{theorem}

We omit the proof of this theorem, as the algorithm is extremely
similar to the one used in the Read-Only/Read-Write model. The main
difference is that the For loop that moves one element of each block
to its correct place is applied before a block is copied on $Y$ and
not after like in Algorithm~\ref{RORWalg}. Similarly, the new
algorithm starts with blocks of size $s$ and ends with size
$\Thet(\sqrt{n/s})$ instead of having blocks of decreasing size.

\smallskip

\textbf{Lower bound.}We employ techniques similar to the ones we used
in the proof of Theorem~\ref{theo_ROWO}. However, here we will not
consider individual cells on the stream but instead blocks of size $k
= \sqrt{ns}$, where $s$ is the memory space. This allows us to easily
bound the amount of information each block receives.

\begin{theorem}\label{theo_RORW_lb}
Let $0 < \eps \leq {1}/{3}$. Every $p$-pass $\lambda$-expansion RO/RW
streaming algorithm 
for $\Reverse$ on $\{0,1\}^n$ with error $\eps$ requires space $\Omeg
(n/{p^2})$. 
\end{theorem}

\begin{proof}
Consider a $p$-pass $\lambda$-expansion randomized RO/RW algorithm for
$\Reverse$ on $\{0,1\}^n$ with error %at most
$\eps$ and space $s$, with $s \geq \log n$.  Like with
Theorem~\ref{theo_ROWO}, we consider the model where we want $Y = X$,
but $Y$ is processed from right to left.

As in the proof of Theorem~\ref{theo_ROWO}, we assume passes are
synchronized at the cost of a factor at most $2$ in $p$. We also keep
similar notations : $X$ is the input stream uniformly distributed in
$\{0,1\}^n$, and for each $1 \leq t \leq p$, $Y^t \in \{0,1,\bot\}^n$
is the data currently on output stream $Y$ at pass $t$. Unlike with
$Z^t$ in the previous section, this includes the data 
previously written, as in this model we can 
read it and modify it. Let $1\leq k\leq n$ be some parameter.
We now think on $X,Y^t$ as sequences of $k$ blocks of size $n/k$, and
consider each block as a symbol. If $\lambda > 1$, then everything
written on the output stream (the only one that can grow) after the
$n$-th bit is considered to be part of the $Y_k^t$.
For instance $X_i$ denotes the $i$-th block of $X$, which is of size $n/k$ bits.
%For $0 \leq i \leq k-1$, we write $X_i$ (resp. $Y^t_i$) the block $X[in/k, \dots, (i+1)n/k-1]$ (resp. $Y^t[in/k+1, \dots, (i+1)n/k]$). 
We write $X_{-i}$ for $X$ without its $i$-th block, and $X_{>i}$
(resp. $X_{<i}$) for the last $(k-i)$ blocks of $X$
(resp. the first $(i-1)$ blocks). 
For each $1 \leq t \leq p$, let $L_t \in\{0,\dots,k-1\}$ be the block where the input head and the output head
meet during the $t$-th pass. Since passes are synchronised, $L_t$ is
unique and is the only block where both heads can be simultaneously
during the $t$-th pass. Let $M_i^t$ be the memory state as the output head
enters the $i$-th block during $t$-th pass.

Consider a pass $t$ and a block $i$. We would like to have an upper
bound on the amount of mutual information between $Y^t_i$ and $X_i$
that is gained during pass $t$ (with regards to information
known at pass $t-1$). Let $\Delta_{i,j}^t =
\mi(X_i:Y^t_i|L_t=j,X_{-i}) - \mi(X_i:Y^{t-1}_i|L_t=j,X_{-i})$
for some $i$ and $j$.  Of course, if $i=j$, without looking
inside the block structure we only have the trivial bound
$\Delta_{i,i}^t \leq H(X_i) = n/k$. It is however easier to bound
other blocks. Assume without loss of generality that $i < j$.
We use the data processing inequality $I(f(A):B|C) \leq
I(A:B|C)$ with $Y^t_i$ as a function of $M_i^t$ and $Y_i^{t-1}$. This gives us
$$\Delta^t_{i,j} \leq \mi(X_i:X_{>i},M_i^t,Y^{t-1}_i|L_t=j,X_{-i}) - \mi(X_i:Y^{t-1}_i|L_t=j,X_{-i}).$$
We can remove $X_{>i}$ which is
contained in the conditioning. Applying the chain rule, we cancel out
the second term and are left with $$\Delta^t_{i,j} \leq
\mi(X_i:M_i^t|L_t=j,X_{-i},Y^{t-1}_i) \leq \h(M_i^t) \leq s.$$
The same holds if $j<i$ instead. If $j=i$, then $\Delta_{i,j}^t \leq n/k$
because $\h(X_i)=n/k$.

We fix $j$, sum over $i$ and get  
$\sum_{i=0}^{k-1}\Delta^t_{i,j} \leq n/k + ks$. The
expectation over $j \sim L_t$ is 
$$\sum_{i=0}^{k-1}\mi(X_i:Y^t_i|L_t,X_{-i})-\mi(X_i:Y^{t-1}_i|L_t,X_{-i}) \leq n/k+ks.$$
From Fact~\ref{I(A:B|C)}, we get the following inequalities:
$$\mi(X_i:Y^t_i|L_t,X_{-i}) \geq \mi(X_i:Y^t_i|X_{-i}) - \h(L_t),$$
$$\mi(X_i:Y^{t-1}_i|L_t,X_{-i}) \leq \mi(X_i:Y^{t-1}_i|X_{-i}) + \h(L_t).$$
Therefore, 
$$\sum_{i=0}^{k-1}\mi(X_i:Y^t_i|X_{-i}) - \mi(X_i:Y^{t-1}_i|X_{-i})
\leq n/k + ks + k\log k.$$ 
Summing over $t$ yields 
$$\sum_{i=0}^{k-1}\mi(X_i:Y^p_i|X_{-i}) \leq p(n/k+ks+k\log k).$$ 
By hypothesis $s\geq\log n$, $\eps\leq 1/3$ and 
$\mi(X_i:Y^p_i|X_{-i}) \geq (1 - \h(\eps))n/k$. 
If $k=\sqrt{n/s}$, it follows that $s = \Omeg(n/{p^2})$.
\end{proof}

\begin{fact}
\label{I(A:B|C)}
Let $A$, $B$, $C$, $D$ be random variables. Then 
$$\mi(A:B|D) -\h(C|D) \leq  \mi(A:B|C,D) \leq \mi(A:B|D) + \h(C|D).$$
\end{fact}

\begin{theorem}\label{theo_RWWO_lb}
Let $0 < \eps \leq {1}/{3}$. Every $p$-pass $\lambda$-expansion RW/WO
streaming algorithm 
for $\Reverse$ on $\{0,1\}^n$ with error $\eps$ requires space 
$\Omeg(n/{p^2})$. 
\end{theorem}

\begin{proof}[Proof of Theorem~\ref{theo_RWWO_lb}]
Consider a $p$-pass $\lambda$-expansion randomized RW/WO algorithm for
$\Reverse$ on $\{0,1\}^n$ with error %at most
$\eps$ and space $s$, with $s \geq \log n$. As before, we consider the
model where the algorithm writes $X$ on $Y$, using passes from right
to left. This proof is similar to the proof of Theorem~\ref{theo_RORW_lb}.

We proceed as before: we assume
passes are synchronized at the cost of a factor at most $2$ in $p$. We
also keep similar notations : $X$ is the input stream uniformly
distributed in $\{0,1\}^n$, and for each $1 \leq t \leq p$, $X^t \in
\{0,1\}^n$ is the content of the input stream and $Y^t \in
\{0,1,\bot\}^n$ is the content of the output stream
$Y$ at pass $t$. Let $1\leq k\leq n$ be some parameter.
As before, whe think of $X,X^t,Y^t$ as sequences of $k$ blocks of size
$n/k$, and $X_i$ denotes the $i$-th block of $X$. If $\lambda > 1$,
then everithing written on the inpu stream after the $n$-th bit is
considered part of $X_k^t$. We write $X_{-i}$,
$X_{>i}$ and $X_{<i}$ as before. We also define $X_i^{\leq t}$ as the
$(t+1)$-uple $(X_i,X_i^1, \dots, X_i^t)$, i.e. the history of the
$i$-th block until pass $t$.

As in prevoius proofs, for each $1 \leq t \leq p$, let $L_t
\in\{0,\dots,k-1\}$ be the block where the input head and the output
head meet during the $t$-th pass. $L_t$ is
unique and is the only block where both heads can be simultaneously
during the $t$-th pass. Let $M_i^t$ be the memory state as the output head
enters the $i$-th block during $t$-th pass.

Consider a pass $t$ and a block $i$. As with
Theorem~\ref{theo_RORW_lb}, we would like to have an upper
bound on the amount of mutual information between $Y^t_i$ and $X_i$
that is gained during pass $t$, assuming $L_t \not = i$.
Let $\Delta_{i,j}^t = \mi(X_i:Y^t_i|L_t=j,X_{-i}^{\leq t-1}) -
\mi(X_i:Y^{t-1}_i|L_t=j,X_{-i}^{\leq t-1})$ for some $j > i$.
By the data processing inequality, we have 
$I(f(A):B|C) \leq I(A:B|C)$. Therefore, 
$$\Delta^t_{i,j} \leq \mi(X_i:X_{>i}^{t-1},M_i^t,Y^{t-1}_i|L_t=j,X_{-i}^{\leq t-1}) - \mi(X_i:Y^{t-1}_i|L_t=j,X_{-i}^{\leq t-1}).$$
We can remove $X_{>i}^{t-1}$ which is
contained in the conditioning. Applying the chain rule, we cancel out
the second term and are left with 
$$\Delta^t_{i,j} \leq \mi(X_i:M_i^t|L_t=j,X_{-i}^{\leq t-1},Y^{t-1}_i) \leq \h(M_i^t) \leq s.$$
The same holds if $j<i$ instead. If $j=i$, then 
$\Delta_{i,j}^t \leq n/k$ because $\h(X_i)=n/k$.

We fix $j$, sum over $i$ and get  
$\sum_{i=0}^{k-1}\Delta^t_{i,j} \leq n/k + ks$. As before, by taking
the expectation over $j \sim L_t$ and then using Fact~\ref{I(A:B|C)},
we can remove the condition $L_t=j$. This gives us 
$$\sum_{i=0}^{k-1}\mi(X_i:Y^t_i|X_{-i}^{\leq t-1}) - \mi(X_i:Y^{t-1}_i|X_{-i}^{\leq t-1}) \leq n/k + ks + k\log k.$$
We cannot sum over $t$ yet because the
conditionning $X_{-i}^{\leq t-1}$ depends on $t$. However, because
$X_{-i}^t$ is a function of $X_{-i}^{t-1}$, the memory state at the
beginning of the pass and the memory as the head on the input tape
leaves the $i$-th block, applying Fact~\ref{I(A:B|C)} again yields 
$$\mi(X_i:Y^t_i|X_{-i}^{\leq t}) - \mi(X_i:Y^t_i|X_{-i}^{\leq t-1}) \leq \h(X_{-i}^t|X_{-i}^{\leq t-1}) \leq 2s.$$
This is a consequence of 
the output stream being write-only, which we had not used until now.

Therefore $\sum_{i=0}^{k-1}\mi(X_i:Y^t_i|X_{-i}^{\leq t-1}) -\mi(X_i:Y^{t-1}_i|X_{-i}^{\leq t-1}) \leq n/k + 3ks + k\log k.$
Summing over $t$ yields 
$$\sum_{i=0}^{k-1}\mi(X_i:Y^p_i|X_{-i}) \leq p(n/k+3ks+k\log k).$$
By hypothesis $s\geq\log n$, $\eps\leq 1/3$ and 
$\mi(X_i:Y^p_i|X_{-i}) \geq (1 - \h(\eps))n/k$. If $k=\sqrt{n/s}$, 
it follows that $s = \Omeg(n/{p^2})$.
\end{proof}

Note that the proof still works if we relax the write-only model by
allowing the algorithm to rewrite over data that was previously written
on the output stream.

\subsection{Complexity of Read-write/Read-write}
Algorithm~\ref{RWRWalg} proceeds by dichotomy. For simplicity, we
assume that $n$ is a power of $2$, but the algorithm can
easily be adapted while keeping $\lambda=1$. At each step, it splits
the input in two, copies one half to its correct place on the stream,
then makes another pass to copy the other half, effectively exchanging
them. 
%The analysis is differed to Appendix~\ref{ap:rwrw}.

\begin{lstlisting}[caption={RW/RW streaming algorithm for $\Reverse$},
label=RWRWalg,captionpos=t,float,abovecaptionskip=-\medskipamount,mathescape]
$W_0 \gets X$; $W_1 \gets Y$; $\alpha \gets 0$; // Rename the streams
$k \gets n$; // Size of current blocks
While $k > 1$
  $k \gets k/2$ (*@\label{l/2}@*)
  For $i=1$ to $n/2k - 1$ (*@\label{pass1}@*)
    $W_{1-\alpha}[(2i+1)k+1,(2i+2)k] \gets W_{\alpha}[2ik+1,(2i+1)k]$ // One pass
  For $i=0$ to $n/2k - 1$  (*@\label{pass2}@*)
    $W_{1-\alpha}[2ik+1,(2i+1)k] \gets  W_{\alpha}[(2i+1)k+1,(2i+2)k]$ // Another pass
  $\alpha \gets 1- \alpha$; Erase $W_{1-\alpha}$; // Exchange the roles of the streams
$Y \gets W_{\alpha}$ // Copy the final result on output tape
\end{lstlisting}

\begin{theorem}
\label{theo_RWRW}
Algorithm~\ref{RWRWalg} is a deterministic
$O(\log n)$-passes RW/RW streaming algorithm for
$\Reverse$ on $\Sigma^n$ with space $\Order(\log n)$ and expansion
$1$.
\end{theorem}
\begin{proof}
Since the algorithm can read and write on both tapes, they perform
very similar roles. We rename the input stream $W_0$ and the output stream
$W_1$.
By a simple recursion, we see that whenever a block
$W_{1-\alpha}[tk,(t+1)k]$ is moved, it is moved in the place that
$\Reverse(W_{1-\alpha}[tk,(t+1)k])$ will occupy. Therefore, the
algorithm is correct.

Now we prove the bounds on $s$ and $p$. Algorithm~\ref{RWRWalg} never
needs to remember a value, only the current index and current pass, so
$s = \Order(\log n)$. Since the length of blocks copied is divided by
$2$ at each execution of line~\ref{l/2}, it ends after a logarithmic
number of executions of the While loop. 
Each iteration of the While loop requires
two passes, one for each of the two For loops (lines~\ref{pass1} and~\ref{pass2}). Therefore the total number of passes is in $\Order(\log n)$.
\end{proof}

\section{Sorting with two streams}
\label{Sort}
$\Sort$ is generally more complex than $\Reverse$. Even in
the RW/RW model, we are not able to present a deterministic algorithm as
efficient as Algorithm~\ref{RWRWalg} for $\Reverse$.
With three streams, the problem becomes easy
since we can Merge Sort two streams, and write the
result of each step on the third one.

\subsection{Merge Sort}

We begin with an algorithm inspired from~\cite{cy91}. 
The algorithm works as a Merge Sort. We call $B_i^t$ the $i$-th sorted
block at the $t$-th iteration of the While loop, consisting of the
sorted values of $X[2^ti+1,2^t(i+1)]$.
Since there is no third stream to write on when two blocks $B_{2i-1}^t$
and $B_{2i}^t$ are merged into $B_i^{t+1}$, 
we label each element with its position in the new block.
Then both halves are copied on the same stream again so that they can be merged with the help of the labels. 
This improves the expansion of~\cite{cy91}  from $n$ to $\log
n$. However it is somewhat unsatisfying because when $\Sigma$ is of
constant size, our algorithm still has $\Omeg(\log n)$ expansion.

\begin{theorem}
\label{theo_MS}
Algorithm~\ref{MS} is a deterministic
  $\Order(\log n)$-pass RW/RW streaming algorithm for
  $\Sort$ on $\Sigma^n$ with space $\Order(\log n)$ and expansion
$ \Order ( {(\log n)}/{\log |\Sigma|} )$.
\end{theorem}

\begin{proof}
Since it is an implementation of the Merge Sort algorithm, Algorithm
\ref{MS} is correct. Each iteration of the While loop corresponds to
five passes on each tape, and therefore the total number of passes is in 
$\Order(\log n)$. 
Since the algorithm only needs to remember the position of the heads,
current elements and the counters $k$, $t$, it only uses memory
$\Order(\log n)$. Finally, since the label for each element is at
most $n$, we only use  space $\log n$ on the stream to write it, and
therefore the expansion is at most $ ({\log|\Sigma| + \log
  n})/{\log |\Sigma|} = \Order (({\log n})/{\log |\Sigma|} )$.
\end{proof}

\begin{lstlisting}[caption={RW/RW streaming algorithm implementing Merge Sort},label=MS,captionpos=t,float,abovecaptionskip=-\medskipamount,mathescape]
$W_0 \gets X$; $W_1 \gets Y$; $\alpha \gets 0$;
$t \gets 1$; $k \gets 1$ // Size of the sorted blocks
Expand the input : each element has a label of size $\log n$.
While $k < n$
  For $i=1$ to $n/2k$ {Copy $B_{2i}^t$ on $W_{1-\alpha}$}
  For $i=1$ to $n/2k$ {For each $W_{\alpha}[j] \in B_{2i-1}^t \cup B_{2i}^t$ 
                       {$W_{\alpha}[j] \gets$ index of $W_{\alpha}[j]$ in $B_i^{t+1}$}}
  For $i=1$ to $n/2k$ {Copy $B_{2i}^t$ at the end of $W_{\alpha}$ after $B_{2i-1}^t$}
  For $i=1$ to $n/2k$ {For each $W_{\alpha}[j] \in B_{2i-1}^t$ 
                       {Write $W_{\alpha}[j]$ at its position on $W_{1-\alpha}$}}
  For $i=1$ to $n/2k$ {For each $W_{\alpha}[j] \in B_{2i}^t$ 
                       {Write $W_{\alpha}[j]$ at its position on $W_{1-\alpha}$}}
  $\alpha \gets 1 - \alpha$; Erase $W_{1- \alpha}$; $t \gets t+1$; $k \gets 2k$;
$Y \gets W_{\alpha}$

\end{lstlisting}

\subsection{Quick Sort}

With a Quick Sort algorithm instead of a Merge Sort, we only need
to store the current pivot (of size at most $\log n$),
without labeling elements. However, Quick Sort comes
with its own issues: the expected number of executions of the While loop 
is $\Order(\log n)$, but unless we can select a good pivot it
is $\Omeg(n)$ in the worst case. For this reason, we use a randomized
Las vegas algorithm.

A block in Algorithm~\ref{QS} is a set of elements that are still
pairwise unsorted, i.e. elements that have the same relative positions
to all pivots so far. The block $B_i^t$ is the $i$-th lower one
during the $t$-th iteration of the While loop, and $P_i^t$ is its
pivot. The block $B_{2i-1}^{t+1}$ consists of all elements in $B_i^t$
lower than $P_i^t$ and all elements equal $P_i^t$ with a lower
index. The block $B_{2i}^{t+1}$ is the complementary. Algorithm
\ref{QS} marks the borders of blocks with the symbol
$\sharp$.

Algorithm~\ref{QS} selects each pivot $P_i^t$ at random among the
elements of $B_i^t$. While it may not do so uniformly with only one
pass because $|B_i^t|$ is unknown, it has an upper bound $k \geq
|B_i^t|$. Algorithm~\ref{QS} selects $l \in \{1,\dots,k\}$ uniformly
at random, then picks $l_i$ the remainder modulo $2^{ \lceil \log
  |B_i^t| \rceil}$ of $l$. This can be computed in one pass with
$\Order(\log n)$ space by updating $l_i$ as the lower bound on
$|B_i^t|$ grows. It uses $P_i^t= B_i^t[l_i]$. $P_i^t$ is selected
uniformly at random from a subset of size at least half of $B_i^t$,
which guarantees that with high probability
$min(|B_{2i-1}^{t+1}|,|B_{2i}^t|) = \Omeg(|B_i^t|)$.

\begin{lstlisting}[caption={RW/RW streaming algorithm implementing Quick Sort},label=QS,captionpos=t,float,abovecaptionskip=-\medskipamount,mathescape]
$W_0 \gets X$; $W_1 \gets Y$; $\alpha \gets 0$;
$t \gets 1$; $K \gets 1$ // Number of unsorted blocks
While $K > 0$
  Abort if the total number of passes is $\Omeg((\log n)/\eps)$
  Expand $W_0$ adding $\Order(K)$ space for $\sharp$ and pivots
  For $i=1$ to $K$
    Find $P_i^t$ at random
    $W_{1-\alpha}[i] \gets P_i^t$
    Replace $P_i^t$ with a $\bot$ on $W_{\alpha}$ 
  For $i=1$ to $K$
    Copy $W_{1-\alpha}[i] = P_i^t$ at the start of $B_i^t$ on $W_{\alpha}$
  For $i=1$ to $K$
    Write all elements in $B_{2i-1}^{t+1}$ on $W_{1-\alpha}$
    Write $\sharp P_i \sharp$ on $W_{1-\alpha}$.
    Leave space for the rest of $B_{2i}^{t+1}$
  For $i=1$ to $K$
    Write all elements in $B_{2i}^{t+1}$ in the space left on $W_{1-\alpha}$
  $\alpha \gets 1 - \alpha$; Erase $W_{1-\alpha}$; $t \gets t+1$
  $K \gets$ new number of unsorted blocks // using an additional pass
$Y \gets W_{\alpha}$
\end{lstlisting}

\begin{theorem}
\label{theo_QS}
%Let $\Sigma$ be any alphabet with some total order.
For all $0<\eps\leq 1/2$, Algorithm~\ref{QS}
is a randomized $\Order((\log n)/\eps)$-pass RW/RW
Las Vegas streaming algorithm for $\Sigma^n$ with failure $\eps$,
space $\Order(\log n)$ and expansion $\Order(1)$.
\end{theorem}

\begin{proof}
Algorithm \ref{QS} is an implementation of the Quick Sort algorithm,
and is therefore correct. Its  correctness does not depend on
the quality of the pivots. The bound on the memory and the expansion
are direct. While it uses additional symbols
$\bot$ and $\sharp$, they can be easily replaced by an encoding with
symbols of $\Sigma$ appearing in the input with only $\Order(1)$
expansion. Because for each $i$ and $t$, with high probability
$min(|B_{2i-1}^{t+1}|,|B_{2i}^t|) = \Omeg(|B_i^t|)$, with high
probability Algorithm~\ref{QS} terminates in $\Order(\log n)$ passes.
\end{proof}

\section*{Open problems}
The first open problem is whether our lower bounds still hold 
if multiple rewrites are allowed in the read-only/write-only model.

Another one is whether there exists a deterministic
$\Order(\log n)$-pass RW/RW streaming algorithm for $\Sort$ 
with space $\Order(\log n)$ and expansion $\Order(1)$. 
We can derandomize Algorithm~\ref{QS} using any deterministic algorithm 
that finds a good approximation of the median.
The best
algorithm we obtained that way has $\Order(\alpha^{-1}\log n)$ passes,
$\Order(n^{\alpha})$ memory and $\Order(1)$ expansion, for any $\alpha
> 0$. We conjecture that there is no such algorithm and that having
constant expansion algorithm for $\Sort$ requires a tradeoff in number
of passes, memory space, number of streams or determinism.

Last, combining the algorithm from Theorem~\ref{theo_RORW_alg} with the
results in~\cite{mmn10}, we obtain a $\Order(\sqrt{n}/\log n)$-pass
RO/RW streaming algorithm with space $\Order((\log n)^2)$ for recognizing
well-parenthesized expressions with two parentheses types.
We do not know if this is optimal.

\bibliographystyle{plain}
\bibliography{fm13}

\end{document}